\documentclass[runningheads]{llncs}
\usepackage{graphicx}
\usepackage{amsmath}
\usepackage{amssymb}
\usepackage[english]{babel}
\usepackage[capitalise]{cleveref}
\usepackage{blindtext}
\usepackage{xcolor}
\usepackage{tikz}

\newcommand{\prop}{\mathsf{Prop}}

\let\sectionsymbol\S

\renewcommand{\phi}{\varphi}

\renewcommand{\S}{\mathsf{S}}
\newcommand{\E}{\mathsf{E}}
\newcommand{\K}{\mathsf{K}}
\newcommand{\univ}{\mathsf{A}}
\newcommand{\univdual}{\hat{\univ}}

\newcommand{\cL}{\mathcal{L}}

\newcommand{\cLKA}{\cL_{\K\univ}}
\newcommand{\cLSA}{\cL_{\S\univ}}
\newcommand{\sL}{\mathsf{L}}
\newcommand{\sLSA}{\sL_{\S\univ}}

\newcommand{\axiom}[1]{{\upshape{(#1)}}}
\newcommand{\axiomsub}[2]{{\upshape{(#1\textsubscript{$#2$})}}}
\newcommand{\Ks}{\axiomsub{K}{\S}}
\newcommand{\Ts}{\axiomsub{T}{\S}}
\newcommand{\fives}{\axiomsub{5}{\S}}
\newcommand{\Kuniv}{\axiomsub{K}{\univ}}
\newcommand{\Tuniv}{\axiomsub{T}{\univ}}
\newcommand{\fiveuniv}{\axiomsub{5}{\univ}}
\newcommand{\es}{\axiom{ES}}
\newcommand{\inc}{\axiom{Inc}}
\renewcommand{\mp}{\axiom{MP}}
\newcommand{\necuniv}{\axiomsub{Nec}{\univ}}
\newcommand{\rs}{\axiomsub{R}{\S}}

\renewcommand{\and}{\wedge}
\newcommand{\orr}{\vee}
\newcommand{\true}{\top}
\newcommand{\falsum}{\bot}
\newcommand{\sat}{\vDash}
\newcommand{\sataug}{\sat^\text{aug}}
\newcommand{\entails}{\vdash}
\newcommand{\entailsL}{\entails_{\sL}}
\newcommand{\entailsSA}{\entails_{\sLSA}}

\newcommand{\identicallyequal}{\equiv}

\begin{document}
\title{A Logic of Expertise}
\author{Joseph Singleton}
\authorrunning{J. Singleton}
\institute{Cardiff University, Cardiff, UK\\
\email{singletonj1@cardiff.ac.uk}}
\maketitle
\begin{abstract}
In this paper we introduce a simple modal logic framework to reason about the
\emph{expertise} of an information source. In the framework, a source is an
expert on a proposition $p$ if they are able to correctly determine the truth
value of $p$ in any possible world. We also consider how information may be
false, but true after \emph{accounting for the lack of expertise} of the
source. This is relevant for modelling situations in which information sources
make claims beyond their domain of expertise.
We use non-standard semantics for the language based on an \emph{expertise set} with
certain closure properties. It turns out there is a close connection between
our semantics and S5 epistemic logic, so that expertise can be expressed in
terms of \emph{knowledge} at all possible states. We use this connection to
obtain a sound and complete axiomatisation.

\end{abstract}
\section{Introduction}

Many scenarios require handling information from non-expert sources. Such
information can be false, even when sources are sincere, when sources make
claims regarding topics on which they are not experts. Accordingly, the (lack of)
expertise of the source must be taken into account when new information is
received.

In this paper we introduce a modal logic formalism for reasoning about the
expertise of information sources. Our logic includes operators for
expertise ($\E\phi$) and \emph{soundness} ($\S\phi$). Intuitively, a source $s$
has expertise on $\phi$ if $s$ is able to correctly determine the truth value
of $\phi$ in \emph{any possible world}. On the other hand, $\phi$ is sound if
it true \emph{up to the limits of the expertise} of $s$. That is, if $\phi$ is
logically weakened to ignore information beyond the expertise of $s$, the resulting
formula is true. This provides a crucial link between expertise
and truthfulness of information, which allows some information to derived from
false statements.

The related notion of \emph{trust} has been well-studied from a logical
perspective~\cite{booth_trust_2018,liau_2003,dastani2004inferring,lorini2014trust,herzig2010logic}.
Despite some similarities, trustworthiness on $\phi$ is different from
expertise on $\phi$. Firstly, whether $s$ is trusted on $\phi$ is a property of
the \emph{truster}, not of $s$. In contrast, we aim to model expertise
objectively as a property of $s$ alone. Secondly, we interpret expertise
\emph{globally}: whether or not $s$ is an expert on $\phi$ does not depend on
the truth value of $\phi$ in any particular state. In this sense expertise is a
counterfactual notion, in that it refers to possible worlds other than the
``actual'' one. This is not necessarily so for trust; e.g. we may not want to trust
the judgement of $s$ on $\phi$ if we know $\phi$ to be false in the actual
world.

\textbf{Contribution and paper outline.} Our main conceptual contribution is a
modal logic framework for reasoning about expertise and soundness of
information. This framework is motivated via an example in
\cref{sec:expertise_and_soundness}, after which the syntax and semantics are
formally introduced. \Cref{sec:s5_link} goes on to establish a connection
between our logic and \emph{S5 epistemic logic}, which provides an alternative
interpretation of our notion of expertise in terms of S5 knowledge. Our main
technical result is a sound and complete axiomatisation, given in
\cref{sec:axiomatisation}.

\section{Expertise and Soundness}
\label{sec:expertise_and_soundness}

The core notions we aim to model are \emph{expertise} and \emph{soundness} of
information. We illustrate both with a simplified example.

\begin{example}
\label{ex:economist_motivation}

Consider an economist reporting on the effects of COVID-19 vaccine rollout, who
states that widespread vaccination will aid economic recovery ($r$), but that
the vaccine can cause health problems ($p$). Assume that the economist is an
expert on matters to do with the economy ($\E r$), so that they only provide
correct information on proposition $r$, but is not an expert on matters of
health ($\neg\E p$). For the sake of the example, suppose economic
recovery will indeed follow, but there are \emph{no} health problems associated
with the vaccine. Then while the economist's report of $r \and p$ is false, it
is true on the propositions on which they are an expert. Consequently, if one
ignores the parts of the report on which the economist has no expertise, the
report becomes true. We say that $r \and p$ is \emph{sound}, given the
expertise of the source on $r$ but not $p$.

\end{example}

\subsection{Syntax}

We introduce the language of expertise and soundness.
Let $\prop$ be a countable set of propositional variables. The language $\cL$
is defined by the following grammar:
\[ \phi ::= p \mid \neg\phi \mid \phi \and \phi \mid \E\phi \mid \S\phi \mid
\univ\phi \]
for $p \in \prop$. Note that formulas of $\cL$ describe the expertise of a
\emph{single} source. The language can be easily extended to handle multiple
sources by adding modalities $\E_s$ and $\S_s$ for each source $s$, but we do
not do so here.

We read $\E\phi$ as ``the source has expertise on $\phi$'',
and $\S\phi$ as ``$\phi$ is sound for the source to report''. We include the
universal modality $\univ$~\cite{goranko_1992} for technical reasons to aid
in the axiomatisation of \cref{sec:axiomatisation}; $\univ\phi$ is to be read
as ``$\phi$ holds in all states''. Other Boolean connectives ($\orr$,
$\rightarrow$, $\leftrightarrow$) and truth values ($\true$, $\falsum$) are
introduced as abbreviations. We denote by $\hat{\E}$, $\hat{\S}$ and
$\univdual$ the dual operators corresponding to $\E$, $\S$ and $\univ$
respectively (e.g. $\hat{\E}\phi$ stands for $\neg \E \neg \phi$).

\subsection{Semantics}

Formulas of $\cL$ are interpreted via non-standard semantics on the basis of an
\emph{expertise set}.

\begin{definition}
An \emph{expertise frame} is a pair $F = (X, P)$, where $X$ is a set of states
and $P \subseteq 2^X$ is an \emph{expertise set} satisfying the following
properties:
\begin{enumerate}
    \item[] \textbf{(P1)} $X \in P$
    \item[] \textbf{(P2)} If $A \in P$ then $X \setminus A \in P$
    \item[] \textbf{(P3)} If $\{A_i\}_{i \in I} \subseteq P$, then $\bigcap_{i
          \in I}A_i \in
          P$
\end{enumerate}

An \emph{expertise model} is a triple $M = (X, P, v)$, where $(X, P)$ is an
expertise frame and $v: \prop \to 2^X$ is a valuation function.

\end{definition}

Intuitively, $A \in P$ means the source has the expertise to determine whether
or not the ``actual'' state of the world, whatever that may be, lies in $A$.
Implicitly in this interpretation we assume that the expertise of the source
does not depend on this ``actual'' state. This makes our semantics for
expertise formulas a special case of the \emph{neighbourhood
semantics}~\cite{pacuit2017neighborhood}, where the neighbourhoods of all
states in $X$ are the same (i.e. $N(x) \identicallyequal P$).

\textbf{(P1)} simply says that the
source is able to determine that the state lies in $X$, i.e. that the source
has expertise on tautologies. \textbf{(P2)} says that $P$ is closed under
complements. This is a natural requirement, given our intended interpretation
of $P$: if the source can determine whether the actual state $x$ lies inside
$A$ or not, the same clearly holds for $X \setminus A$. Note that
\textbf{(P1)} and \textbf{(P2)} together imply that $\emptyset \in P$.
Finally, \textbf{(P3)} says that $P$ is closed under (arbitrary) intersections,
which implies expertise is closed under conjunctions. Together with
\textbf{(P2)}, this implies $P$ is also closed under (arbitrary) unions, and
thus expertise is also closed under disjunctions.
We come to the truth conditions for $\cL$ formulas with respect to expertise
models.

\begin{definition}
\label{def:truth_conditions}
Let $M = (X, P, v)$ be an expertise model. The satisfaction relation between
points $x \in X$ and formulas $\phi \in \cL$ is defined inductively as follows:
\[
    \begin{array}{lll}
        &M, x \sat p &\iff x \in v(p) \\
        &M, x \sat \neg\phi &\iff M, x \not\sat \phi \\
        &M, x \sat \phi \and \psi &\iff M, x \sat \phi
            \text{ and } M, x \sat \psi \\
        &M, x \sat \E\phi &\iff \|\phi\|_M \in P \\
        &M, x \sat \S\phi &\iff \text{ for all } A \in P,\ \|\phi\|_M \subseteq A
            \text{ implies } x \in A \\
        &M, x \sat \univ\phi &\iff \text{ for all } y \in X,\ M, y \sat \phi
    \end{array}
\]
where $\|\phi\|_M = \{x \in X \mid M, x \sat \phi\}$. We write $M \sat \phi$ if
$M, x \sat \phi$ for all $x \in X$, and $\sat \phi$ if $M \sat \phi$ for all
expertise models $M$; we say $\phi$ is \emph{valid} in this case. Write
$\phi \equiv \psi$ if $\sat \phi \leftrightarrow \psi$.

\end{definition}

The clauses for propositional variables and propositional connectives are
standard, and the clause for $\univ\phi$ is straightforward. The clause for
$\E\phi$ follows the intuition highlighted above: $\E\phi$ is true iff the set
of states $\|\phi\|_M$ at which $\phi$ is true lies in the expertise
set. Note that the truth value of $\E\phi$ does not depend on
the state $x$. For $\S\phi$ to hold at $x \in X$, we require that
all supersets of $\|\phi\|_M$ on which the source is an expert must
contain $x$. That is, any logical weakening of $\phi$ is true at $x$, whenever
the source has expertise on the weaker formula. We illustrate the semantics by
formalising \cref{ex:economist_motivation}, and conclude this section by noting
some (in)validities that follow directly from \cref{def:truth_conditions}.

\begin{example}
\label{ex:economist_formalised}

Consider a model $M = (X, P, v)$, where $X = \{a,b,c,d\}$,
$v(r) = \{a,c\}$ and $v(p) = \{a, b\}$, and $P = \{\emptyset, \{a,c\}, \{b,d\},
X\}$. Then $a$ satisfies $r \and p$, $b$ satisfies $\neg r \and p$, $c$
satisfies $r \and \neg p$, and $d$ satisfies $\neg r \and \neg p$.

In \cref{ex:economist_motivation} we assumed the economist had expertise on
$r$. Here we have $\|r\|_M = \{a,c\} \in P$, so $M \sat \E r$ as expected. We
also claimed $r \and p$ was sound when $r$ is true and $p$ is false, i.e. $M, c
\sat \S(r \and p)$. Indeed, $\|r \and p\|_M = \{a\}$, and the supersets of
$\{a\}$ in $P$ are $\{a,c\}$ and $X$. Clearly both contain $c$, so $M, c \sat
\S(r \and p)$. This situation is also depicted graphically in
\cref{fig:graphical_example}.

\end{example}

\begin{figure}
    \centering
    \begin{tikzpicture}[scale=3]
        \node (a) at (0, 0) {\Large $a$};
        \node (b) at (1, 0) {\Large $b$};
        \node (c) at (0, 1) {\Large $c$};
        \node (d) at (1, 1) {\Large $d$};

        \draw[rounded corners,thick] (-0.2, -0.2) rectangle (1.2, 0.2);
        \node at (0.5, 0) {$\|p\|_M$};
        \draw[rounded corners,thick] (-0.2, -0.2) rectangle (0.2, 1.2);
        \node at (0, 0.5) {$\|r\|_M$};
        \draw[rounded corners,draw=red,thick] (-0.17, -0.17) rectangle (0.17, 1.17);
        \draw[rounded corners,draw=red,thick] (1 - 0.17, -0.17) rectangle (1 + 0.17, 1.17);
        \draw[rounded corners,draw=blue,thick] (-0.15, -0.15) rectangle (0.15, 0.15);
    \end{tikzpicture}
    \caption{
        Graphical depiction of the model in \cref{ex:economist_formalised}. The
        red boxes show the sets in $P$ (with $\emptyset$ and $X$ omitted).
        While the report $r \and p$ is false at $c$ (i.e. $c$ is not contained
        in the blue box), all red boxes containing the blue box do contain $c$,
        so $M, c \sat \S(r \and p)$.
    }
    \label{fig:graphical_example}
\end{figure}
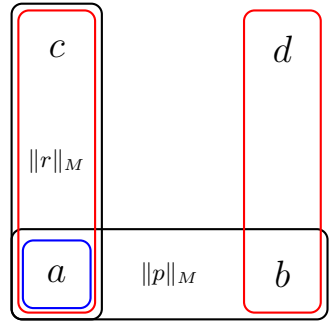

\begin{proposition}
\label{prop:basic_validities}
For any $\phi, \psi \in \cL$ and any model $M$,
\begin{enumerate}
    \item $\E\phi \equiv \E\neg\phi \equiv \univ\E\phi$
    \item Either $M \sat \E\phi$ or $M \sat \neg\E\phi$
    \item $\sat \E\true \and \E\falsum \and \E\E\phi$
    \item $\sat (\E\phi \and \E\psi) \rightarrow \E(\phi \and \psi)$
    \item The distribution axiom $\E(\phi \rightarrow \psi) \rightarrow (\E\phi
          \rightarrow \E\psi)$ is not in general valid\footnotemark{}
    \item $\sat \phi \rightarrow \S\phi$
    \item If $\sat \phi \rightarrow \psi$ then $\sat (\S\phi \and \E\psi)
          \rightarrow \psi$
\end{enumerate}

\footnotetext{
    For a counterexample, consider $X = \{a,b,c\}$, $P =
    \{\emptyset, \{a\}, \{b,c\}, X\}$ and $v$ such that $v(p) = \{a\}$, $v(q) =
    \{b\}$. Then $\|p\|_M = \{a\} \in P$, $\|q\|_M = \{b\} \notin P$ and $\|p
    \rightarrow q\|_M = \{b,c\} \in P$. Then we have $M \sat \E(p \rightarrow
    q) \and \E p \and \neg\E q$.
}

\end{proposition}

\section{Connection with S5 Epistemic Logic}
\label{sec:s5_link}

In this section we show that, despite the non-standard semantics for $\E\phi$
and $\S\phi$, expertise and soundness can be equivalently defined by the
standard relational semantics~\cite{blackburn2002modal} in the language $\cLKA$
with knowledge and universal modalities $\K$ and $\univ$. It will be seen
that the accessibility relation for $\K$ in the relational model $M^*$
corresponding to an expertise model $M$ is in fact an equivalence relation, so
that $M^*$ is an \emph{S5 model}~\cite[\sectionsymbol 4.1]{blackburn2002modal}.
S5 represents an ``ideal'' form of knowledge, which satisfies the KT5 axioms:
knowledge is closed under logical consequence (K: $\K(\phi \rightarrow \psi)
\rightarrow (\K\phi \rightarrow \K\psi)$), all that is known is true (T:
$\K\phi \rightarrow \phi$), and if $\phi$ is not known, this is itself known
(5: $\neg\K\phi \rightarrow \K\neg\K\phi$).
First, let us define the relational semantics for $\cLKA$.

\begin{definition}
\label{def:relational_semantics}
A \emph{relational model} is a triple $M' = (X, R, v)$, where $X$ is a set of
states, $R \subseteq X \times X$ is a binary relation on $X$, and $v: \prop \to
2^X$ is a valuation. Given a relational model $M'$, the satisfaction relation
between points $x \in X$ and formulas $\phi \in \cLKA$ is defined inductively
by
\[
    \begin{array}{lll}
        &M', x \sat \K\phi &\iff \text{ for all } y \in X, xRy
            \text { implies } M', y \sat \phi \\
        &M', x \sat \univ\phi &\iff \text{ for all } y \in X,\ M', y \sat \phi
    \end{array}
\]
where the clauses for propositional connectives are as in
\cref{def:truth_conditions}.
\end{definition}

Say a relational model $M' = (X, R, v)$ is an \emph{S5 model} if $R$ is an
equivalence relation. In the context of S5, $R$ is an \emph{epistemic
accessibility relation}: $xRy$ means the source considers $y$ as a possible
state if the actual state is $x$. The source then `knows' $\phi$ at $x$ if
$\phi$ is true in every state accessible from $x$.

The following result shows how one can form a unique S5 model from an expertise
model, and vice versa.

\begin{lemma}
\label{lemma:expertise_set_equiv_reln}
Let $X$ be a set. Let $\mathcal{P}$ denote the set of expertise sets over $X$
-- i.e.  the set of all $P \subseteq X$ satisfying \textbf{(P1)}, \textbf{(P2)}
and \textbf{(P3)} -- and let $\mathcal{E}$ denote the set of equivalence
relations over $X$. Then there is a bijection $P \mapsto R_P$ from
$\mathcal{P}$ into $\mathcal{E}$ such that, for all $A \subseteq X$,
\begin{equation}
\label{eqn:p_unions_of_equiv_classes}
    A \in P \iff A \text{ is a union of equivalence classes of } R_P
\end{equation}
\end{lemma}

\begin{proof}
    Given an expertise set  $P \in \mathcal{P}$ and $x \in X$, write
    \[
        A_x = \bigcap \{A \in P \mid x \in A\}
    \]
    Note that $A_x \in P$ by \textbf{(P3)}, so $A_x$ is the smallest set in $P$
    containing $x$.\footnote{Also note that $X \in P$ by \textbf{(P1)}, so
    there is at least one $A \in P$ containing $x$.} Set $\Pi = \{A_x \mid x
    \in X\} \subseteq P$. We claim that $\Pi$ is a partition of $X$. It is
    clear that $\Pi$ covers $X$, since each $x$ lies in $A_x$ by definition. We
    show that any distinct $A_x, A_y \in \Pi$ are disjoint. Without loss of
    generality, $A_x \not\subseteq A_y$. Hence $x \notin A_y$ (otherwise $A_y$
    appears in the intersection defining $A_x$ and we get $A_x \subseteq A_y$).
    That is, $x \in A_x \setminus A_y = A_x \cap (X \setminus A_y)$. But this
    difference lies in $P$ by \textbf{(P2)} and \textbf{(P3}). Since $A_x$ is
    the smallest set in $P$ containing $x$, we get $A_x \subseteq A_x \setminus
    A_y$. In particular, $A_x \cap A_y = \emptyset$.

    Let $R_P$ be the equivalence relation defined by the partition $\Pi$, i.e.
    $x R_P y$ iff $A_x = A_y$. We show \cref{eqn:p_unions_of_equiv_classes}
    holds. First suppose $A \in P$. Then $A = \bigcup_{x \in A}A_x$; the
    left-to-right inclusion is clear since $x \in A_x$ for all $x$, and the
    right-to-left inclusion holds since $x \in A$ implies $A_x \subseteq
    A$ for $A \in P$. Since the $A_x$ form the equivalence classes of $R_P$, we
    are done.

    Now suppose $A$ is a union of equivalence classes of $R_P$, i.e. $A =
    \bigcup_{x \in B}{A_x}$ for some $B \subseteq X$. Since each $A_x$ lies in
    $P$ and $P$ is closed under unions by \textbf{(P2)} and \textbf{(P3)}, we
    have $A \in P$. Hence \cref{eqn:p_unions_of_equiv_classes} is shown.

    It only remains to show that the mapping $P \mapsto R_P$ is bijective.
    Injectivity follows easily from \cref{eqn:p_unions_of_equiv_classes}, since
    $P$ is fully determined by $R_P$. For surjectivity, take any equivalence
    relation $R \in \mathcal{E}$ on $X$. For $x \in X$, let $[x]_R$ denote the
    equivalence class of $X$. Let $P$ consist of all unions of equivalence
    classes, i.e.
    \[
        P = \left\{\bigcup_{x \in B}{[x]_R} \mid B \subseteq X\right\}
    \]
    We need to show that $P \in \mathcal{P}$ -- i.e. \textbf{(P1)},
    \textbf{(P2)} and \textbf{(P3)} hold -- and that $R_P = R$. For
    \textbf{(P1)}, taking $B = X$ gives $X \in P$. For \textbf{(P2)}, suppose
    $A = \bigcup_{x \in B}{[x]_R} \in P$. It is easily verified that $X
    \setminus A = \bigcup_{y \in X \setminus A}{[y]_R} \in P$, so \textbf{(P2)}
    holds. \textbf{(P3)} follows from \textbf{(P2)} and the fact that $P$ is
    closed under unions, which is evident from the definition. Finally, it
    follows from the definition of $P$ and \cref{eqn:p_unions_of_equiv_classes}
    that a set $A \subseteq X$ is a union of equivalence classes of $R$ if and
    only if it is a union of equivalence classes of $R_P$. Since distinct
    equivalence classes are disjoint, this implies that the equivalence classes
    of $R$ coincide with those of $R_P$, and $R = R_P$ as required. \qed
\end{proof}

On the syntactic side, define a translation $t: \cL \to \cLKA$ inductively by
$t(p) = p$, $t(\neg\phi) = \neg t(\phi)$, $t(\phi \and \psi) = t(\phi) \and
t(\psi)$, $t(\univ\phi) = \univ t(\phi)$, and
\[
        t(\E\phi) = \univ(t(\phi) \rightarrow \K t(\phi));
        \quad
        \quad
        t(\S\phi) = \neg\K\neg t(\phi)
\]
We then have that $\phi \in \cL$ is true in an expertise model exactly when
when $t(\phi)$ is true in the induced S5 model $M^*$.

\begin{theorem}
\label{thm:s5_link}
Let $M = (X, P, v)$ be an expertise model. Then $M^* = (X, R_P, v)$ is an S5
model, and
\[
    M, x \sat \phi \iff M^*, x \sat t(\phi)
\]
\end{theorem}

Before the proof, note that since the mapping $P \mapsto R_P$ is a bijection
into the set of equivalence relations on $X$ (by
\cref{lemma:expertise_set_equiv_reln}), any S5 model $M' = (X, R, v)$ has an
expertise counterpart $M = (X, P, v)$ such that $M^* = M'$. In this sense, the
converse of \cref{thm:s5_link} also holds.

\begin{proof}[\cref{thm:s5_link}]

Let $M = (X, P, v)$ be an expertise model. By
\cref{lemma:expertise_set_equiv_reln}, $R_P$ is an equivalence relation and
$M^*$ is indeed an S5 model. Let $\Pi$ denote the partition of $X$
corresponding to $R_P$, and as in \cref{lemma:expertise_set_equiv_reln}, let
$A_x \in \Pi$ denote the cell of $\Pi$ containing $x$, i.e. the equivalence
class of $x$ in $R_P$. By \cref{eqn:p_unions_of_equiv_classes} in
\cref{lemma:expertise_set_equiv_reln}, $A \in P$ iff $A$ is a union of cells
from $\Pi$.

We show the desired semantic correspondence by induction on formulas. The cases
for the Boolean connectives and $\univ$ are straightforward. Suppose the result
holds for $\phi$ and $M, x \sat \E\phi$. Then $\|\phi\|_M \in P$, so
$\|\phi\|_M = \bigcup\mathcal{A}$ for some collection $\mathcal{A} \subseteq
\Pi$. Now suppose $y \in X$ and $M^*, y \sat t(\phi)$. Suppose $y R_P z$. By the
inductive hypothesis, $\|t(\phi)\|_{M^*} = \|\phi\|_M$, so $y \in
\bigcup\mathcal{A}$.  Since $A_y$ is the unique set in $\Pi$ containing $y$, we
must have $A_y \in \mathcal{A}$. Consequently, $y R_P z$ implies $z \in A_y
\subseteq \bigcup\mathcal{A} = \|t(\phi)\|_{M^*}$. That is, $M^*, z \sat
t(\phi)$. This shows $M^*, y \sat t(\phi) \rightarrow \K t(\phi)$ for arbitrary
$y \in X$, and so $M^*, x \sat \univ(t(\phi) \rightarrow \K t(\phi))$, i.e.
$M^*, x \sat t(\E\phi)$.

Conversely, suppose $M^*, x \sat \univ(t(\phi) \rightarrow \K t(\phi))$. We
claim that $\|\phi\|_M = \bigcup_{y \in \|\phi\|_M}A_y$. The left-to-right
inclusion is clear since $y \in A_y$ for each $y$. For the reverse inclusion,
let $y \in \|\phi\|_M$ and $z \in A_y$. Then $y R_P z$. By the inductive
hypothesis, $M^*, y \sat t(\phi)$. Since $t(\phi) \rightarrow \K t(\phi)$ holds
everywhere in $M^*$ by assumption, we have $M^*, y \sat \K t(\phi)$. Hence
$M^*, z \sat t(\phi)$, so $M, z \sat \phi$ and $z \in \|\phi\|_M$. This shows
$\|\phi\|_M = \bigcup_{y \in \|\phi\|_M}A_y$, i.e. $\|\phi\|_M$ is a union of
cells of $\Pi$. Hence $\|\phi\|_M \in P$ and $M, x \sat \E\phi$ as required.

Next we take the $\S\phi$ case. We prove both directions by contraposition.
First suppose $M, x \not\sat \S\phi$. Then there is some $A \in P$ with
$\|\phi\|_M \subseteq A$ and $x \notin A$. Suppose $x R_P y$. Then $A_x = A_y$. If
$y \in A$ we would get $A_x = A_y \subseteq A$, since $A_y$ is the smallest set
in $P$ containing $y$, but this contradicts $x \notin
A$. Hence $y \notin A$. In particular, $y \notin \|\phi\|_M$. By the inductive
hypothesis, $y \notin \|t(\phi)\|_{M^*}$, so $M^*, y \sat \neg t(\phi)$. This
shows $M^*, x \sat \K \neg t(\phi)$, i.e. $M^*,x \not\sat \neg\K\neg t(\phi)$
as required.

Finally, suppose $M^*,x \not\sat \neg\K\neg t(\phi)$. Take $A = \bigcup_{y \in
\|\phi\|_M}A_y$. Since each $A_y$ is in $P$ and $P$ is closed under unions by
\textbf{(P2}) and \textbf{(P3)}, we have $A \in P$. Clearly $\|\phi\|_M
\subseteq A$. Suppose for contradiction that $M, x \sat \S\phi$. Then $x \in
A$, i.e. there is $y \in \|\phi\|_M$ such that $x \in A_y$. Consequently $x R_P y$,
and $M^*, x \sat \K\neg t(\phi)$ implies $M^*, y \sat \neg t(\phi)$. But this
means $y \notin \|t(\phi)\|_{M^*} = \|\phi\|_M$ -- contradiction. \qed

\end{proof}

Note that, in the case of a propositional formula $\phi$, the translation $t$
takes $\S \phi$ to $\neg\K\neg \phi$, and $\E \phi$ to $\univ(\phi \rightarrow
\K \phi)$. The semantic correspondence in \Cref{thm:s5_link} therefore shows
that the soundness operator $\S$ is just the dual of an S5 knowledge operator:
$\phi$ is sound iff the source does not \emph{know} $\neg\phi$. Similarly,
$\E\phi$ holds iff for \emph{all} possible states, if $\phi$ were true then the
source would know it. Moreover, the equivalence relation used to interpret $\K$
is uniquely derived from the expertise model which interprets $\E$ and $\S$, by
\cref{lemma:expertise_set_equiv_reln}. This gives a new interpretation of
expertise and soundness which refers directly to the source's epistemic state
via the $\K$ operator.

\Cref{thm:s5_link} also allows $\E\phi$ be expressed solely in terms
of $\univ$ and $\S$:
\[
    \E\phi \equiv \univ(\S\phi \rightarrow \phi)
\]
i.e. the source has expertise on $\phi$ iff, in every possible
state, $\phi$ is sound only if it is in fact true. This can be seen by
recalling that $\E\phi$ is equivalent to $\E\neg\phi$ (by
\cref{prop:basic_validities}), and noting that $t(\E\neg\phi)$ is equivalent to
$t(\univ(\S\phi \rightarrow \phi))$.
Similarly, we can \emph{lack} of expertise in terms of $\S$ and the dual
operator $\univdual$:
\[
    \neg\E\phi \equiv \univdual(\S\phi \and \neg\phi) \equiv \univdual(\phi \and
    \S\neg\phi)
\]

\section{Axiomatisation}
\label{sec:axiomatisation}

\Cref{thm:s5_link} demonstrates a close semantic link between the logic of
expertise and S5. Accordingly, we can obtain a sound and complete
axiomatisation of the validities of $\cL$ by adapting any axiomatisation of S5
(although some care is required to handle the universal modality). Let $\sL$ be
the extension of the propositional calculus containing the axioms and inference
rules shows in \cref{tab:axiomatisation}.

\begin{table}
    \centering
    \caption{Axioms and inference rules for $\sL$.}
    \label{tab:axiomatisation}
    \begin{tabular}{ll}
        \hline
            \Ks      & $\S\phi \and \neg\S\psi \rightarrow \S(\phi \and \neg\psi)$ \\
            \Ts      & $\phi \rightarrow \S\phi$ \\
            \fives   & $\S\neg\S\phi \rightarrow \neg\S\phi$ \\
            \Kuniv & $\univ(\phi \rightarrow \psi) \rightarrow
                        (\univ\phi \rightarrow \univ\psi)$ \\
            \Tuniv & $\univ\phi \rightarrow \phi$ \\
            \fiveuniv & $\neg\univ\phi \rightarrow \univ\neg\univ\phi$ \\
            \es         & $\E\phi \leftrightarrow \univ(\S\phi \rightarrow \phi)$ \\
            \inc        & $\univ\phi \rightarrow \neg\S\neg\phi$ \\
            \mp         & From $\phi$ and $\phi \rightarrow \psi$ infer $\psi$ \\
            \necuniv  & From $\phi$ infer $\univ\phi$ \\
            \rs         & From $\phi \leftrightarrow \psi$ infer $\S\phi \leftrightarrow \S\psi$ \\
        \hline
    \end{tabular}
\end{table}

Here \Kuniv{}, \Tuniv{} and \fiveuniv{} are the standard KT5 axioms for
$\univ$, which characterise S5. \Ks{}, \Ts{} and \fives{} are reformulations
of the KT5 axioms for the dual operator $\hat\S = \neg\S\neg$; we present them
in terms of $\S$ rather than $\hat\S$ to aid readability and intuitive
interpretation of the axioms. \es{} captures the interaction between expertise
and soundness; the validity of this axiom was already shown as a consequence of
\cref{thm:s5_link}. Note that the necessitation rule \necuniv{} for $\univ$
and \inc{} imply necessitation for $\hat\S$ by \mp{}.

\begin{theorem}
\label{thm:axiomatisation}

$\sL$ is sound\footnote{Soundness of the logic $\sL$ for expertise frames
should not be confused with the notion of soundness inside the language.} and
complete with respect to expertise frames.

\end{theorem}

Soundness is immediate for most of the axioms and inference rules. We give
details only for \Ks{} and \fives{}.

\begin{lemma}
    Axioms \Ks{} and \fives{} are valid in all expertise frames.
\end{lemma}

\begin{proof}
Let $M = (X, P, v)$ be an expertise model. For \Ks{}, suppose $M, x \sat \S\phi
\and \neg\S\psi$. Take any $A \in P$ with $\|\phi \and \neg \psi\|_M \subseteq
A$. Then $\|\phi\|_M \setminus \|\psi\|_M \subseteq A$, so $\|\phi\|_M
\subseteq A \cup \|\psi\|_M$. Since $M, x \not\sat \S\psi$, there is $B \in P$
with $\|\psi\|_M \subseteq B$ and $x \notin B$. Now, $\|\phi\|_M \subseteq A
\cup \|\psi\|_M \subseteq A \cup B$, and $A \cup B \in P$ since $P$ is closed
under unions. Since $M, x \sat \S\phi$, any superset of $\|\phi\|_M$ in $P$
must contain $x$. Hence $x \in A \cup B$. But $x \notin B$, so we must have $x
\in A$. This shows $M, x \sat \S(\phi \and \neg\psi)$.

For \fives{}, suppose $M, x \sat \S\neg\S\phi$. It can be seen from
\cref{def:truth_conditions} that $\|\neg\S\phi\|_M = \bigcup\{X \setminus B
\mid B \in P, \|\phi\|_M \subseteq B\}$. It follows from \textbf{(P2)} and
\textbf{(P3)} that $\|\neg\S\phi\|_M \in P$. Consequently, $\|\neg\S\phi\|_M$
is itself a set in $P$ containing $\|\neg\S\phi\|_M$. Since $M, x \sat
\S\neg\S\phi$ we get $x \in \|\neg\S\phi\|_M$, i.e. $M, x \sat \neg\S\phi$ as
required. \qed

\end{proof}

The completeness proof requires some more machinery, and we use ideas found
in~\cite{bonanno2005simple,goranko_1992}. Let $\cLSA$ denote the fragment of
$\cL$ without the $\E$ modality, and let $\sLSA$ denote the logic of $\sL$ for
$\cLSA$ without axiom \es. For a frame $F = (X, P)$, let $R_P$ denote the
corresponding equivalence relation on $X$ from
\cref{lemma:expertise_set_equiv_reln}. An \emph{augmented expertise frame} is
obtained by adding to any frame $F$ an equivalence relation $R_\univ$ on $X$
such that $R_P \subseteq R_\univ$ (c.f.~\cite{bonanno2005simple}). An
augmented model $N$ is an augmented frame equipped with a valuation $v$. We
define a satisfaction relation $\sataug$ between augmented models and $\cLSA$
formulas, where \[ N, x \sataug \univ\phi \iff \text{ for all } y \in X,
x{R_\univ}y \text{ implies } N, y \sataug \phi \] and satisfaction for other
formulas is as in \cref{def:truth_conditions}. That is, $\univ\phi$ is no
longer the universal modality, and is instead interpreted via
relational semantics.

\begin{lemma}
\label{lemma:completeness_augmented}
$\sLSA$ is complete for $\cLSA$ with respect to augmented frames.
\end{lemma}

\begin{proof}[sketch]
First note that from \Ks{}, \Ts{}, \fives{} and \rs{}, one can prove as
theorems of $\sLSA$ the usual KT5 axioms for the dual operator $\hat\S$ -- that is, $\entailsSA
\hat\S(\phi \rightarrow \psi) \rightarrow (\hat\S\phi \rightarrow \hat\S\psi)$,
$\entailsSA \hat\S\phi \rightarrow \psi$ and $\entailsSA \neg\hat\S\phi
\rightarrow \hat\S\neg\hat\S\phi$ -- where $\hat\S$ is an abbreviation for
$\neg\S\neg$. As remarked before, we also have the necessitation rule for both
$\hat\S$ and $\univ$ by \necuniv{} and \inc{}.

By the standard canonical model construction~\cite{blackburn2002modal}, we
obtain the canonical relational model $(X, R_{\hat\S}, R_{\univ}, v)$, where
$X$ is the set of all maximally $\sLSA$-consistent subsets (MCS) of $\cLSA$,
$R_{\hat\S}$ and $R_\univ$ are accessibility relations for $\hat\S$ and
$\univ$ respectively, and $\Delta \in X$ satisfies $\phi$ under the
relational semantics iff $\phi \in \Delta$, for any $\phi \in \cLSA$ (this fact
is known as the \emph{truth lemma}). Moreover, $R_{\hat\S}$ and $R_{\univ}$
are equivalence relations by the KT5 axioms for $\hat\S$ and $\univ$
respectively, and \inc{} implies $R_{\hat\S} \subseteq R_\univ$.
By \cref{lemma:expertise_set_equiv_reln}, there is an expertise set $P$ such
that $R_P = R_{\hat{S}}$.
Consequently, we
obtain an augmented model $N = (X, P, R_{\univ}, v)$. Applying the link
between expertise-based and relational semantics established in
\cref{thm:s5_link}, one can adapt the truth lemma to show that $N, \Delta
\sataug \phi$ iff $\phi \in \Delta$ for any MCS $\Delta \in X$ and $\phi \in
\cLSA$. Completeness now follows by contraposition. If $\phi \in \cLSA$ is not
a theorem of $\sLSA$, then $\{\neg\phi\}$ is $\sLSA$-consistent, and so there
is a MCS $\Delta$ containing $\neg\phi$ by Lindenbaum's
Lemma~\cite{blackburn2002modal}. Consequently $\phi \notin \Delta$, so $N,
\Delta \not\sataug \phi$ and $\phi$ is not valid in augmented frames. \qed

\end{proof}

Completeness of $\sLSA$ for (non-augmented) expertise frames follows by
considering \emph{generated sub-frames} of augmented frames.

\begin{lemma}
\label{lemma:completeness_lsa}
$\sLSA$ is complete for $\cLSA$ with respect to expertise frames.
\end{lemma}

\begin{proof}[sketch]
Suppose $\phi \in \cLSA$ is not a theorem of $\sLSA$. By
\cref{lemma:completeness_augmented}, there is an augmented model $N = (X, P,
R_\univ, v)$ and a state $x \in X$ such that $N, x \not\sataug \phi$. Let
$X' \subseteq X$ be the equivalence class of $x$ in $R_\univ$. Consider the
generated sub-model $N' = (X', P', R'_\univ, v')$, where $P' = \{A \cap X'
\mid A \in P\}$, $R'_\univ = R_\univ \cap (X' \times X')$ and $v'(p) = v(p)
\cap X'$. It can be shown that for all $\psi \in \cLSA$ and $y \in X'$, we have
$N, y \sataug \psi$ iff $N', y \sataug \psi$.\footnotemark{} Hence $N', x
\not\sataug \phi$.

Now, note that the relation $R'_\univ$ was obtained by restricting
$R_\univ$ to one of its equivalence classes $X'$. It follows that
$R'_\univ$ is in fact the universal relation $X' \times X'$ on $X'$.
Consequently, $N', y \sataug \univ\psi$ iff $N', z \sataug \phi$ for all $z
\in X'$, i.e. $\univ$ is just the universal modality for $N'$.

Writing $M$ for the non-augmented model obtained from $N'$ by simply dropping
the $R'_\univ$ component, we see that $M, y \sat \psi$ iff $N', y \sataug
\psi$, for all $y \in X'$ and $\psi \in \cLSA$. In particular, $M, x \not\sat
\phi$, so $\phi$ is not valid in all expertise frames. \qed

\footnotetext{
    This is clear by induction on formulas, except for the case $\S\psi$. Here
    we use the fact that $R_P \subseteq R_\univ$ to show $X' = \bigcup_{z \in
    X'}[z]_{R_P}$ -- where $[z]_{R_P}$ is the equivalence class of $z$ in $R_P$
    -- which implies $X' \in P$. Using the inductive hypothesis it is then
    straightforward to show that $N, y \sataug \S\psi$ iff $N', y \sataug
    \S\psi$.
}
\end{proof}

The completeness of $\sL$ for the validities of the whole language $\cL$ now
follows. Indeed, let $g: \cL \to \cLSA$ be the natural embedding of $\cL$ in
$\cLSA$, where $g(\E\phi) = \univ(\S g(\phi) \rightarrow
g(\phi))$.\footnote{\ldots and $g(p) = p$, $g(\neg\phi) = \neg g(\phi)$,
$g(\phi \and \psi) = g(\phi) \and g(\psi)$, $g(\S\phi) = \S g(\phi)$ and
$g(\univ\phi) = \univ g(\phi)$.} In light of earlier remarks and axiom \es{},
we have both $\phi \equiv g(\phi)$ and $\entailsL \phi \leftrightarrow g(\phi)$
for all $\phi \in \cL$.  Consequently, $\sat \phi$ implies $\sat g(\phi)$ and
thus $\entailsSA g(\phi)$ by \cref{lemma:completeness_lsa}; since $\sL$ extends
$\sLSA$ we have $\entailsL g(\phi)$, and $\entailsL g(\phi) \rightarrow \phi$
implies $\entailsL \phi$ by \mp{}. This shows (weak) completeness, and
\cref{thm:axiomatisation} is proved.

\section{Conclusion}
\label{sec:conclusion}

This paper introduced a simple modal language to reason about the
expertise of an information source. We used the notion of ``soundness'' -- when
information is true after ignoring parts on which the source has no expertise
-- to establish a connection with S5 epistemic logic. This provided
alternative interpretation of expertise, and led to a sound and complete
axiomatisation.

There are many possible directions for future work. For instance, it may be
unrealistic to expect the expertise set of a source is fully known up front.
Methods for \emph{estimating} the expertise, e.g. based on past
reports~\cite{dastani2004inferring}, could be developed to reason about
expertise approximately in practical settings. The ``binary'' notion of
expertise we consider may also be unrealistic: either $\E\phi$ holds or
$\neg\E\phi$ holds. Enriching the language and semantics to handle
\emph{graded} or \emph{probabilistic} levels is a natural generalisation which
would allow a more nuanced discussion of expertise.

One could also investigate the relation between expertise and \emph{trust}.
For example, can the trustworthiness of a source on $\phi$ be \emph{derived}
from expertise on $\phi$? The language $\cL$ could be extended with a trust
operator $\mathsf{T}$ to model this formally in future work.

Also note that in this paper we only consider \emph{static} expertise. In
reality, expertise may change over time as new evidence becomes available and
as the epistemic state of the information source evolves. One could introduce
\emph{dynamic operators}, as is done in Dynamic Epistemic Logic, to model this
change in expertise in response to evidence and other epistemic events. When it
comes to the interaction between expertise and evidence specifically,
\emph{evidence logics}~\cite{van2011dynamic,vanbenthem2014106} may be highly
relevant. These logics use neighbourhood semantics to interpret the evidence
modalities, which is technically (and perhaps also conceptually) similar to our
semantics for the expertise modality. We save the detailed analysis and
comparison for future work.

\section*{Acknowledgements}

We thank the anonymous ESSLLI 2021 student session reviewers, whose insightful
comments and suggestions have greatly improved the paper.

\bibliographystyle{splncs04}
\bibliography{references}

\begin{thebibliography}{10}
\providecommand{\url}[1]{\texttt{#1}}
\providecommand{\urlprefix}{URL }
\providecommand{\doi}[1]{https://doi.org/#1}

\bibitem{blackburn2002modal}
Blackburn, P., Rijke, M.d., Venema, Y.: Modal Logic. Cambridge University Press
  (2001)

\bibitem{bonanno2005simple}
Bonanno, G.: A simple modal logic for belief revision. Synthese
  \textbf{147}(2),  193--228 (2005)

\bibitem{booth_trust_2018}
Booth, R., Hunter, A.: Trust as a {Precursor} to {Belief} {Revision}. JAIR
  \textbf{61},  699--722 (2018)

\bibitem{dastani2004inferring}
Dastani, M., Herzig, A., Hulstijn, J., Van Der~Torre, L.: Inferring trust. In:
  CLIMA. pp. 144--160. Springer (2004)

\bibitem{goranko_1992}
Goranko, V., Passy, S.: Using the universal modality: Gains and questions.
  Journal of Logic and Computation  \textbf{2}(1),  5--30 (1992)

\bibitem{herzig2010logic}
Herzig, A., Lorini, E., H{\"u}bner, J.F., Vercouter, L.: A logic of trust and
  reputation. Logic Journal of the IGPL  \textbf{18}(1),  214--244 (2010)

\bibitem{liau_2003}
Liau, C.J.: Belief, information acquisition, and trust in multi-agent systems
  -- a modal logic formulation. Artificial Intelligence  \textbf{149}(1),
  31--60 (2003)

\bibitem{lorini2014trust}
Lorini, E., Jiang, G., Perrussel, L.: Trust-based belief change. In: Proc.
  ECAI. pp. 549--554 (2014)

\bibitem{pacuit2017neighborhood}
Pacuit, E.: Neighborhood semantics for modal logic. Springer International
  Publishing (2017)

\bibitem{vanbenthem2014106}
{van Benthem}, J., Fern{\'{a}}ndez-Duque, D., Pacuit, E.: Evidence and
  plausibility in neighborhood structures. Annals of Pure and Applied Logic
  \textbf{165}(1),  106--133 (2014)

\bibitem{van2011dynamic}
{van Benthem}, J., Pacuit, E.: Dynamic logics of evidence-based beliefs. Studia
  Logica  \textbf{99}(1),  61--92 (2011)

\end{thebibliography}
\end{document}